\documentclass[11pt]{amsart}
\usepackage{amsmath,amssymb,amsthm}
\usepackage{graphicx}
\usepackage{booktabs}
\usepackage{hyperref}
\usepackage{microtype}
\usepackage{enumitem}

\newtheorem{theorem}{Theorem}
\newtheorem{proposition}[theorem]{Proposition}

\newtheorem{observation}[theorem]{Observation}
\theoremstyle{definition}

\theoremstyle{remark}
\newtheorem{remark}[theorem]{Remark}

\DeclareMathOperator{\perm}{perm}

\title{Permanents of matrix ensembles:\\computation, distribution, and geometry}

\author{Igor Rivin}
\address{Mathematics Department, Temple University}
\email{rivin@temple.edu}
\date{\today}
\thanks{The author would like to thank Scott Aronson for helpful comments, and Klas Markstrom for pointing out his paper on closely related subject. The author would like to thank Claude Code for the enormous help in bringing these ideas to life}

\begin{document}

\begin{abstract}
We report on a computational and experimental study of permanents.  On the computational side, we use the GPU to greaatly accelerate the computation of permanents over $\mathbb{C},$ $\mathbb{R},$ $\mathbb{F}_p$ and $\mathbb{Q}.$ 
First, for Haar-distributed unitary matrices~$U$, the permanent
$\perm(U)$ follows a circularly-symmetric complex Gaussian
distribution $\mathcal{CN}(0,\sigma^2)$ --- we confirm this via
a number of tests
for $n$ up to~23 with $50{,}000$ samples.  The DFT matrix permanent is
an extreme outlier for every prime $n\ge 7$.  In
contrast, for Haar-random \emph{orthogonal} matrices~$O$, the permanent
$\perm(O)$ is approximately real Gaussian but with positive excess
kurtosis that decays as~$O(1/n)$, indicating slower convergence.
For matrices with Gaussian entries (GUE, GOE, Ginibre),  the permanent follows an $\alpha$-stable
distribution with stability index $\alpha\approx 1.0$--$1.4$,
well below the Gaussian value $\alpha=2$.  We test Aaronson's
conjecture that $|\perm(X)|^2$ is asymptotically lognormal for
Gaussian~$X$: it is plausible for the complex Ginibre and GOE
ensembles, but appears to fail for GUE and real Ginibre, where the
$\alpha$-stable tails prevent convergence.  Anti-concentration, however,
holds for all Gaussian ensembles and is more robust than for Haar
unitaries.

Secondly, we study the permanent along geodesics on the unitary group.  For the geodesic from the identity to the $n$-cycle
permutation matrix, we find a universal scaling function
$f(t)=\frac{1}{n}\ln|\perm(\gamma(t))|$ that is independent of~$n$
in the large-$n$ limit, with a midpoint value
\[
  \perm(\gamma({\textstyle\frac12}))
  = (-1)^{(n-1)/2}\cdot 2e^{-n}\bigl(1+\tfrac{1}{3n}+O(n^{-2})\bigr)
\]
for odd~$n$ and zero for even~$n$.
We also study the geodesic forom the identity to the DFT matrix.
\end{abstract}

\maketitle

\section{Introduction}

The permanent of an $n\times n$ matrix $A=(a_{ij})$,
\[
  \perm(A)=\sum_{\sigma\in S_n}\prod_{i=1}^n a_{i,\sigma(i)},
\]
is one of the most studied and least tractable objects in algebraic
combinatorics.  Unlike the determinant, no polynomial-time algorithm
is known (or expected) for computing the permanent of a general
matrix; the problem is \#P-complete even over~$\{0,1\}$
\cite{Valiant79}.  The fastest known general algorithm, due to Ryser
\cite{Ryser63}, runs in $O(n\,2^n)$ time.

Despite this worst-case complexity, the permanent admits rich
structure in special families.  A particularly natural family consists
of the \emph{unitary} matrices $U\in U(n)$.  Here the permanent has
gained prominence through \emph{boson sampling}: Aaronson and
Arkhipov~\cite{AaronsonArkhipov13} showed that the output
probabilities of an $n$-photon linear optical network are given by
$|\perm(U_S)|^2$, where $U_S$ is an $n\times n$ submatrix of the
network's unitary transfer matrix.  Their hardness proof---and with it,
a leading proposal for demonstrating quantum computational
advantage---hinges on the assumption that permanents of Haar-random
unitaries anti-concentrate, formalized as the Permanent
Anti-Concentration Conjecture (PACC)
\cite{AaronsonArkhipov13,FeffermanGhoshZhan24}.  Experimental
implementations using photonic circuits
\cite{Zhong2020,Crespi2016} have now reached regimes where classical
simulation of the output distribution is computationally challenging.

The \emph{Schur matrix}
$S_n=(e^{2\pi i jk/n})_{j,k=0}^{n-1}$, whose normalized version
$S_n/\sqrt{n}$ is the discrete Fourier transform, plays a
distinguished role: its permanent defines the integer sequence
OEIS~A003112, and its magnitude as a function of~$n$ remains poorly
understood.  The DFT also features in \emph{suppression laws}
for multi-photon interference: Tichy~\cite{Tichy2014} showed that
certain submatrix permanents of the DFT vanish, governing which
transition amplitudes are suppressed in Fourier multi-port
interferometers.

\medskip

In this paper we report three sets of results.

\begin{enumerate}[leftmargin=2em]
\item \textbf{Computation.} We develop a pipeline combining the
  Chinese Remainder Theorem (CRT) with Gray-code-parallelized Ryser
  evaluation on GPU, achieving a $200$--$400\times$ speedup over
  single-threaded CPU.  This lets us extend the known values of
  $\perm(S_n)$ from $n=35$ to $n=43$.

\item \textbf{Distribution.}  We sample permanents of Haar-distributed
  random unitary matrices and find that $\perm(U)$ follows a
  circularly-symmetric complex Gaussian distribution.  The DFT matrix
  is an extreme outlier for prime~$n$.  For orthogonal matrices,
  $\perm(O)$ is approximately real Gaussian but with excess kurtosis
  that decays as $O(1/n)$.  In contrast, for matrices drawn from
  Gaussian ensembles (GUE, GOE, Ginibre), the permanent follows an
  $\alpha$-stable distribution with $\alpha\approx 1.0$--$1.4$,
  indicating heavy tails and infinite variance.

\item \textbf{Geodesics.}  We study how the permanent evolves along
  geodesics on $U(n)$.  The geodesic from~$I$ to the $n$-cycle
  permutation matrix yields a universal scaling function with a clean
  midpoint asymptotic.  The geodesic to the DFT matrix distinguishes
  primes from composites.

\item \textbf{Lognormal conjecture.}  We test Aaronson's conjecture
  that $|\perm(X)|^2$ is asymptotically lognormal for Gaussian~$X$.
  The conjecture is plausible for complex Ginibre and GOE, but appears
  to fail for GUE and real Ginibre.  Anti-concentration holds---and is
  stronger than for Haar unitaries---across all Gaussian ensembles.
\end{enumerate}

\medskip\noindent\textbf{Note.}
When this paper was already on arXiv, Markstr\"om alerted us to the
paper~\cite{LundowMarkstrom19}, for which we are grateful.
In particular, we adopted their suggestion of Kahan compensated
summation~\cite{Kahan65} in the Ryser inner loop, which eliminates the
$O(2^n)$ accumulation error that would otherwise arise.
From a computational standpoint, Lundow and Markstr\"om computed the
permanent of a single real $54\times 54$ matrix using a 400-node
cluster (11{,}200~CPU cores) with mixed double/quadruple-precision
arithmetic; our GPU implementation, using Glynn's formula with Kahan
summation in double precision, can match this on a single NVIDIA~GH200
in comparable wall-clock time.
On the statistical side, the two studies are complementary but differ
in scope: Lundow and Markstr\"om study the distribution of
$|\perm(A)|$ for Ginibre, circular, and Bernoulli matrices via local
analysis near the origin (CDF scaling $F(x)\sim 6x^2-7x^3$, moment
fitting), whereas we identify the \emph{global} distributional forms
--- complex Gaussian for Haar unitaries, $\alpha$-stable for Gaussian
ensembles --- and test them across the full range of the distribution.
In particular, the local scaling $F(x)\sim cx^2$ near the origin is
consistent with both a Rayleigh distribution ($\alpha=2$) and the
heavier-tailed $\alpha$-stable laws ($\alpha<2$) that we find for
Gaussian ensembles, and cannot distinguish between them.  Our
identification of the full distributional form resolves this ambiguity.

\section{Computational methods}\label{sec:computation}

\subsection{CRT-based exact computation}

The Schur matrix $S_n$ has entries $\omega^{jk}$ where
$\omega=e^{2\pi i/n}$.  When $p\equiv 1\pmod{n}$, the group
$(\mathbb{Z}/p\mathbb{Z})^\times$ contains a primitive $n$th root of
unity~$g$, and we may compute $\perm(S_n)\bmod p$ by evaluating the
permanent of the integer matrix $(g^{jk}\bmod p)$ using Ryser's
algorithm over~$\mathbb{Z}/p\mathbb{Z}$.  Combining sufficiently many
such primes via the Chinese Remainder Theorem recovers the exact
integer $\perm(S_n)$.

The number of primes needed is determined by the bound
$|\perm(S_n)|\le n^{n/2}$, which follows from the classical inequality
$|\perm(U)|\le 1$ for unitary~$U$ \cite{Minc78} applied to the
normalized DFT matrix $S_n/\sqrt{n}$.
For $n=43$ this requires $12$~primes, each needing an independent
Ryser evaluation over $2^{43}\approx 8.8\times 10^{12}$ subsets.

\subsection{Gray code block parallelization}

Following \cite{BipedalF2}, we parallelize Ryser's formula by
partitioning the $2^n$ subsets into contiguous Gray code blocks, one
per thread.  Each thread initializes its block's starting row-sum
vector in $O(n)$ time, then iterates through the block with $O(1)$
work per subset (a single column addition or subtraction, followed by
a product accumulation).  This achieves near-linear parallel scaling
and is the key to effective GPU utilization.

\subsection{GPU acceleration}

We implement the Gray-code-block Ryser kernel in CUDA, distributing
$\sim\!10^6$ blocks across the streaming multiprocessors of an NVIDIA
H100 GPU.  All arithmetic is performed in 64-bit integers.  The GPU
achieves $11$--$13\times$ speedup over a 64-core CPU for the mod-$p$
kernel, and $20$--$30\times$ for a specialized mod-$3$ kernel using
the bipedal $\mathbb{F}_2$ encoding of~\cite{BipedalF2}. Our computations of the exact integer values of the Schur permanent  $4.2$, $23$, $88$, and
$424$~minutes for $n=37,39,41,43$ respectively, consistent with the
expected $\sim 4\times$ scaling per increment of~$2$ in~$n$.

%

\section{Distribution of $|\perm(U)|$ for Haar unitaries}
\label{sec:distribution}

\subsection{Experimental setup}

For each odd prime $n\in\{7,11,13,17,19,23\}$, we generate $50{,}000$
independent Haar-distributed unitary matrices $U\in U(n)$ (via QR
decomposition of complex Gaussian matrices with diagonal phase
correction \cite{Mezzadri07}) and compute $\perm(U)$ using a
GPU-batched Ryser kernel (one CUDA thread per matrix).

\subsection{Complex Gaussian distribution}

\begin{observation}\label{obs:rayleigh}
For Haar-distributed $U\in U(n)$ with $n\ge 5$, the permanent
$\perm(U)$ is well-described by a circularly-symmetric complex
Gaussian $\mathcal{CN}(0,\sigma^2)$; equivalently, $|\perm(U)|$
follows a Rayleigh distribution with parameter~$\sigma(n)$.
\end{observation}

Recall that if $Z\sim\mathcal{CN}(0,\sigma^2)$, then
$|Z|$ has the Rayleigh distribution with scale $\sigma/\sqrt{2}$,
so the magnitude and complex distributional statements are
equivalent.  Moreover, since
$|Z|^2=\operatorname{Re}(Z)^2+\operatorname{Im}(Z)^2$ is the sum of
two independent $N(0,\sigma^2\!/2)$ squares, the \emph{squared
amplitude} is exponentially distributed:
\begin{equation}\label{eq:exp}
  |\perm(U)|^2 \;\sim\; \mathrm{Exp}\!\bigl(\text{mean}\;\sigma^2\bigr),
  \qquad
  \Pr\bigl[|\perm(U)|^2 \ge x\bigr] = e^{-x/\sigma^2},
\end{equation}
with $\sigma^2=n!/n^n$.  This is the complex analog of the
Porter--Thomas distribution~\cite{PorterThomas56}---the universal
squared-amplitude statistics for eigenstates of random Hamiltonians in
quantum chaos.  In the boson sampling context (Section~\ref{sec:boson}),
the output probabilities of a Haar-random linear optical network are
proportional to $|\perm(U_S)|^2$, so an exponential law for the
squared permanent directly governs the statistics of photon detection
events.

We test the complex Gaussian hypothesis in two complementary ways.

\medskip\noindent\textbf{Magnitude (Rayleigh).}
Figure~\ref{fig:qq} presents Q--Q plots of $|\perm(U)|$ against the
fitted Rayleigh distribution for $n=7,11,13,17,19,23$
($50{,}000$~samples each).  The agreement is excellent: Kolmogorov--Smirnov
$p$-values range from $0.23$ to $0.92$, and Weibull shape parameters
satisfy $c\in[1.98,\,2.03]$ (the Rayleigh distribution is the special
case $c=2$).

\begin{figure}[ht]
\centering
\includegraphics[width=\textwidth]{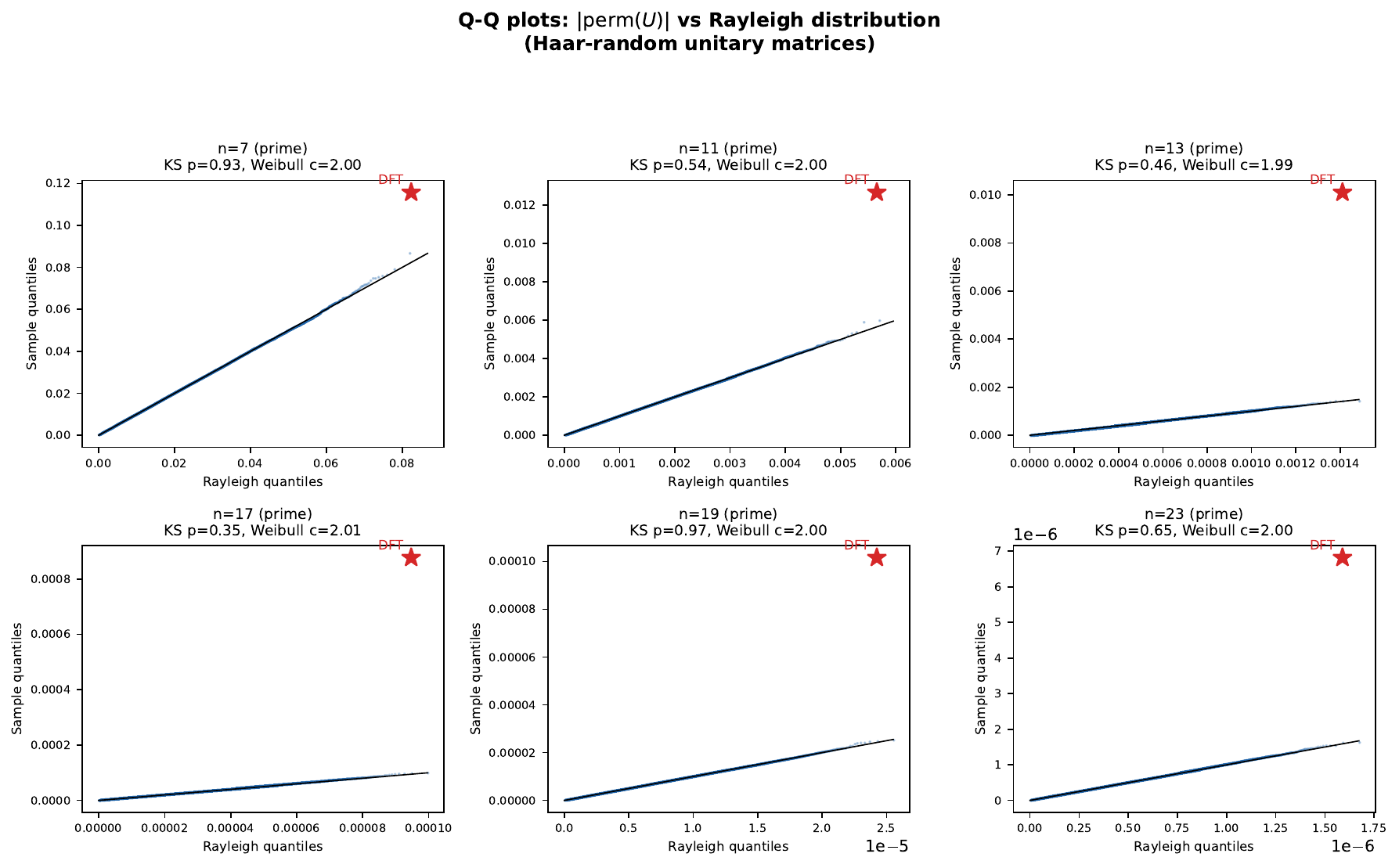}
\caption{Q--Q plots of $|\perm(U)|$ against the Rayleigh distribution
for Haar-random unitary matrices, $n=7,11,13,17,19,23$ ($50{,}000$~samples
each).  The red star marks the DFT matrix, which lies beyond the
100th percentile for every~$n$ shown.  Weibull shape parameters~$c$
and KS $p$-values are shown in each panel title.}\label{fig:qq}
\end{figure}

\medskip\noindent\textbf{Complex structure.}
Figure~\ref{fig:cg} tests the full complex Gaussian hypothesis via
three diagnostics.  The scatter plots of
$(\operatorname{Re}\perm(U),\,\operatorname{Im}\perm(U))$ are
circularly symmetric, with variance ratio
$\sigma_R/\sigma_I\in[0.97,\,1.01]$ and Pearson correlation
$|r|<0.01$.  The Mahalanobis distances $d_k^2$ follow $\chi^2_2$
(Q--Q plots in the middle row), confirming bivariate normality via
Mardia's skewness test ($p>0.03$ for all~$n$, with most $p>0.1$).
The phase $\arg(\perm(U))$ is uniform on $[0,2\pi)$ (bottom row;
KS~$p>0.23$ for all~$n$).  Finally, circularity
($\mathbb{E}[Z^2]=0$) is confirmed by a bootstrap test of the
pseudocovariance, with $|\hat p|/\hat\sigma^2<0.01$ in all cases.

\begin{figure}[ht]
\centering
\includegraphics[width=\textwidth]{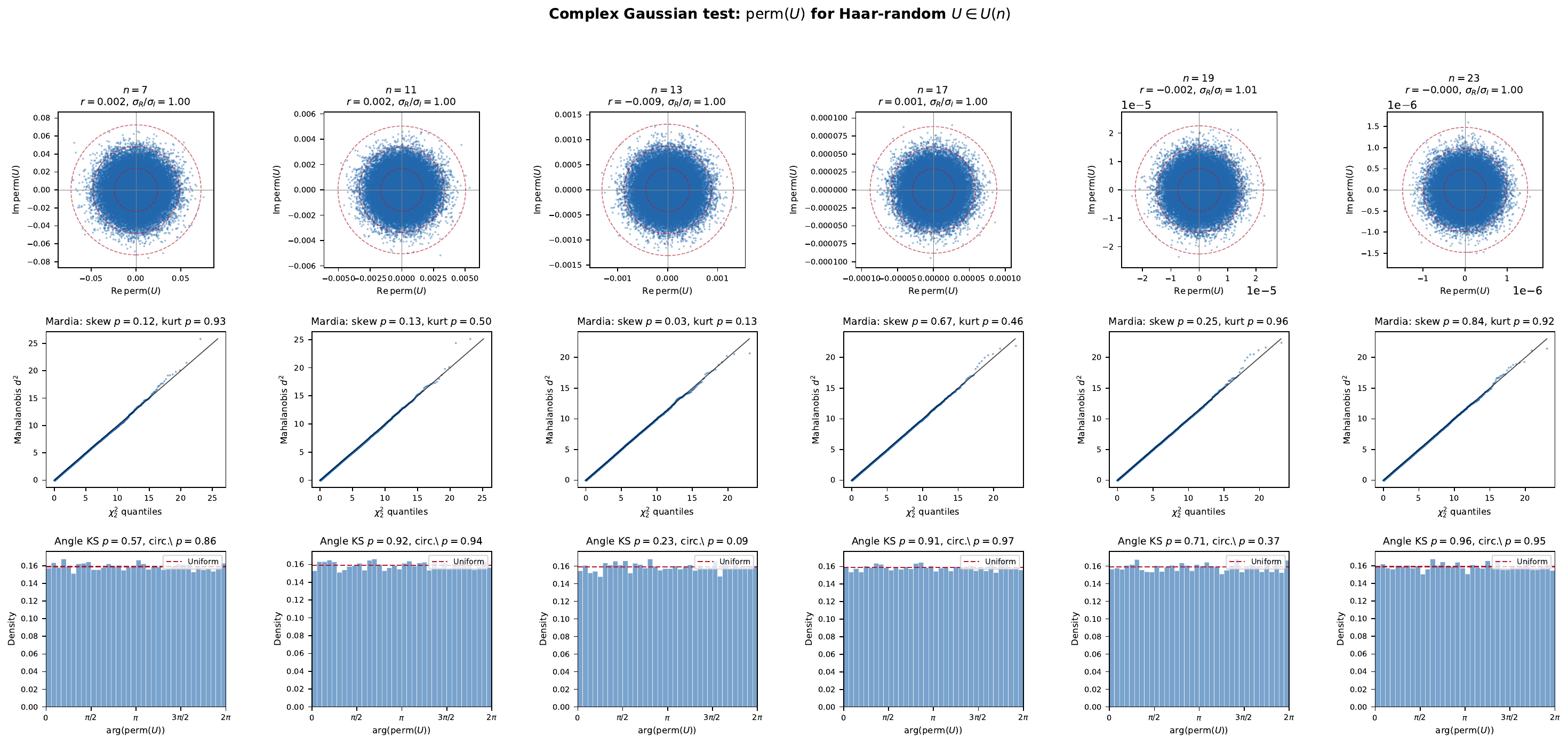}
\caption{Complex Gaussian tests for $\perm(U)$ with Haar-random
$U\in U(n)$, $n=7,11,13,17,19,23$ ($50{,}000$~samples each).
Top: scatter plots of $(\mathrm{Re},\mathrm{Im})$ with
$1\sigma$--$3\sigma$ circles.
Middle: Q--Q plots of Mahalanobis $d^2$ against $\chi^2_2$.
Bottom: phase histograms against the uniform
density.}\label{fig:cg}
\end{figure}

The fitted variance
$\hat\sigma^2(n)\approx n!/n^n\sim\sqrt{2\pi n}\,e^{-n}$
matches the exact second moment
$\mathbb{E}[|\perm(U)|^2]=n!/n^n$ computed via Weingarten calculus
\cite{Nezami21} to within a few percent.  The exponential factor
$e^{-n}$ governs the ``typical size'' of a Haar-random permanent;
equivalently, $|\perm(U)|\sim e^{-n/2}$ with high probability.

A heuristic justification: the permanent is a sum of $n!$ complex
terms $\prod_i U_{i,\sigma(i)}$, each of magnitude $\sim n^{-n/2}$.
For large~$n$, by a central-limit-type argument, this sum approaches a
complex Gaussian, whose magnitude is Rayleigh.  Making this rigorous
would require establishing sufficient decorrelation among the
summands, which we leave as an open problem.

\begin{remark}
While the complex Gaussian form is broadly consistent with the
Porter-Thomas heuristic from quantum chaos (output amplitudes of
random quantum processes are complex Gaussian) and with the CLT
intuition that the permanent---a sum of $n!$ weakly correlated
terms---should be approximately Gaussian, we are not aware of any
prior work that explicitly identifies or tests this distributional
form for $\perm(U)$.  The anti-concentration literature
(e.g.\ \cite{Nezami21,FeffermanGhoshZhan24}) studies moments and tail
bounds but does not identify the full distribution.
\end{remark}

\subsection{DFT outlier for prime $n$}

\begin{observation}\label{obs:dft}
For every odd prime $n$ with $7\le n\le 29$, the DFT permanent
$|\perm(S_n/\sqrt{n})|$ exceeds all $50{,}000$ Haar-random samples,
placing it at the 100th percentile.  For composite odd~$n$ (such as
$n=9$ or $n=15$), the DFT permanent falls within the bulk of the
distribution.
\end{observation}

This dichotomy is summarized in Table~\ref{tab:dft}.

\begin{table}[ht]
\centering
\caption{DFT percentile and ratio to maximum sampled value.}
\label{tab:dft}
\medskip
\begin{tabular}{rccc}
\toprule
$n$ & Type & DFT percentile & $|\perm(F_n)|/\max_{\text{sample}}$ \\
\midrule
7  & prime     & $100\%$ & $1.6\times$ \\
9  & composite & $34\%$ & --- \\
11 & prime     & $100\%$ & $2.5\times$ \\
13 & prime     & $100\%$ & $6.9\times$ \\
15 & composite & $14\%$ & --- \\
17 & prime     & $100\%$ & $10.0\times$ \\
19 & prime     & $100\%$ & $16.0\times$ \\
23 & prime     & $100\%$ & $38.1\times$ \\
\bottomrule
\end{tabular}
\end{table}

The ratio $|\perm(F_n)|/\max_{\text{sample}}$ grows rapidly with~$n$
for primes, indicating that the DFT is not merely in the tail but
\emph{far beyond} the natural range of the Rayleigh distribution.

\section{Permanents along geodesics on $U(n)$}\label{sec:geodesic}

The unitary group $U(n)$ carries a natural bi-invariant Riemannian
metric.  Given a unitary matrix~$U$ with Schur decomposition
$U=Q\,\mathrm{diag}(e^{i\theta_1},\dots,e^{i\theta_n})\,Q^*$,
the geodesic from $I$ to $U$ is
\[
  \gamma(t) = Q\,\mathrm{diag}(e^{it\theta_1},\dots,e^{it\theta_n})\,Q^*,
  \qquad t\in[0,1].
\]
We study $|\perm(\gamma(t))|$ as a function of~$t$ for two natural
choices of endpoint.

\subsection{Geodesic to the $n$-cycle}\label{sec:cycle}

Let $C_n$ be the cyclic permutation matrix $(C_n)_{ij}=\delta_{i,j+1\bmod n}$.
Its eigenvalues are $\omega^k=e^{2\pi ik/n}$ for $k=0,\dots,n-1$,
and it is diagonalized by the DFT matrix.  The geodesic
$\gamma(t)$ is therefore a circulant matrix with explicit entries
\begin{equation}\label{eq:circulant}
  \gamma(t)_{jl}
  = \frac{e^{2\pi it}-1}{n\bigl(e^{2\pi i(l-j+t)/n}-1\bigr)},
\end{equation}
having magnitude $\sin(\pi t)\big/\bigl(n\,|\sin(\pi(l{-}j{+}t)/n)|\bigr)$.

\begin{observation}[Universal scaling]\label{obs:universal}
Define \[f(t)=-\frac{1}{n}\ln|\perm(\gamma(t)).|\].  Then $f(t)$
converges to a universal function of~$t$ as $n\to\infty$, with the
properties:
\begin{enumerate}[label=(\roman*)]
\item Symmetry: $f(t)=f(1-t)$.
\item Boundary: $f(0)=f(1)=0$.
\item Minimum: $f(\tfrac12)=1$ (i.e.\ $|\perm(\gamma(\tfrac12))|\sim e^{-n}$).
\item Gaussian onset: $f(t)\approx \frac{\pi^2}{3}t^2$ for $t\ll 1$.
\end{enumerate}
\end{observation}

Figure~\ref{fig:geodesic} (left panel) shows the collapse of $f(t)$
across $n=5,7,\dots,23$.  The convergence to a single curve is
evident.

\begin{figure}[ht]
\centering
\includegraphics[width=\textwidth]{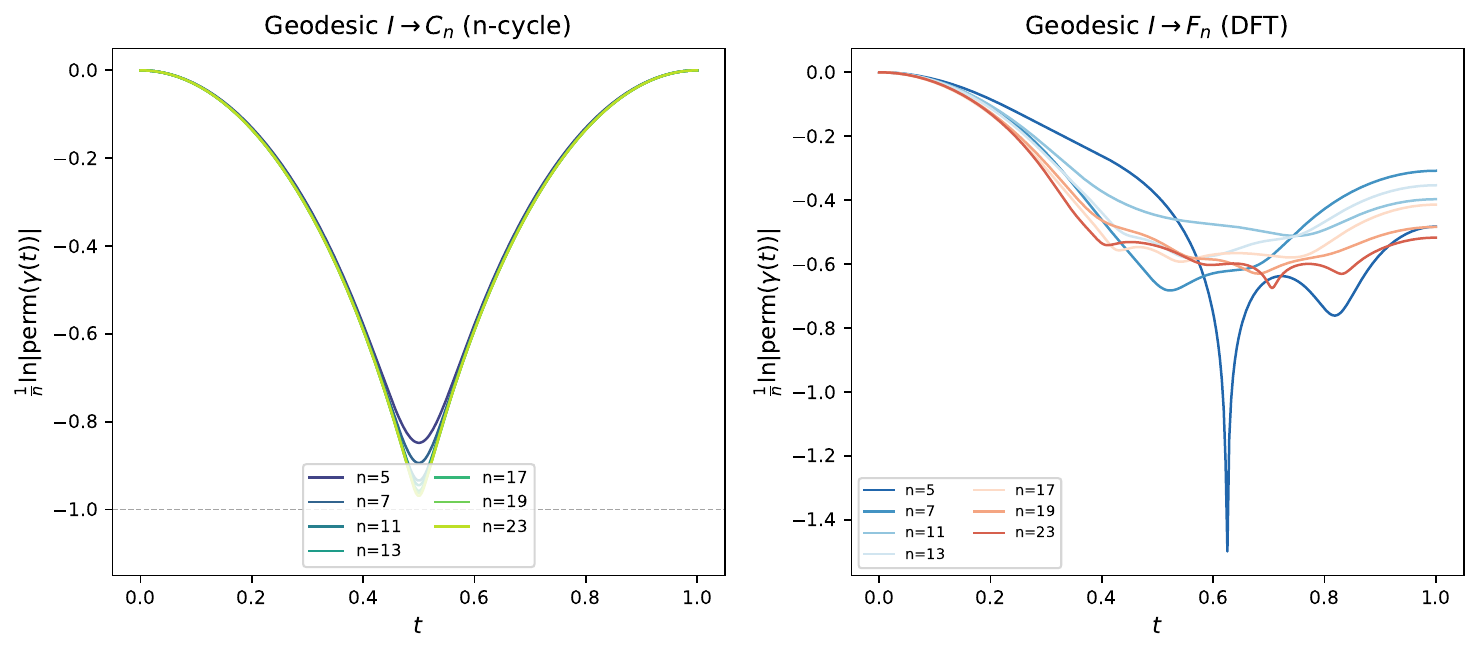}
\caption{Scaled permanent $\frac{1}{n}\ln|\perm(\gamma(t))|$ along
geodesics on $U(n)$.  Left: geodesic $I\to C_n$ showing universal
collapse; the dashed line marks $-1$.  Right: geodesic $I\to F_n$
(DFT) for various primes.}\label{fig:geodesic}
\end{figure}

\begin{observation}[Midpoint asymptotic]\label{obs:midpoint}
For the midpoint $\gamma(\tfrac12)$ of the geodesic $I\to C_n$:
\begin{enumerate}[label=(\roman*)]
\item If $n$ is odd: \quad
  $\perm(\gamma(\tfrac12))
   = (-1)^{(n-1)/2}\cdot 2e^{-n}\bigl(1+\tfrac{1}{3n}+O(n^{-2})\bigr)$.
\item If $n$ is even: \quad
  $\perm(\gamma(\tfrac12))=0$.
\end{enumerate}
\end{observation}

Table~\ref{tab:midpoint} presents the numerical evidence.  The ratio
$|\perm(\gamma(\tfrac12))|/(2e^{-n})$ converges to~$1$ with
correction term $1/(3n)$ verified to four-digit accuracy.

\begin{table}[ht]
\centering
\caption{Midpoint permanent of the $n$-cycle geodesic.}\label{tab:midpoint}
\medskip
\begin{tabular}{rccc}
\toprule
$n$ & $\perm(\gamma(\frac12))$ & $|\perm|/(2e^{-n})$ & $n\cdot(|\perm|/(2e^{-n})-1)$ \\
\midrule
5  & $+1.440\times 10^{-2}$ & $1.0686$ & $0.343$ \\
7  & $-1.912\times 10^{-3}$ & $1.0486$ & $0.340$ \\
11 & $-3.443\times 10^{-5}$ & $1.0307$ & $0.338$ \\
13 & $+4.638\times 10^{-6}$ & $1.0260$ & $0.337$ \\
17 & $+8.444\times 10^{-8}$ & $1.0198$ & $0.336$ \\
19 & $-1.140\times 10^{-8}$ & $1.0177$ & $0.336$ \\
23 & $-2.082\times 10^{-10}$ & $1.0146$ & $0.335$ \\
\bottomrule
\end{tabular}
\end{table}

\begin{remark}[Phase cancellation]
The permanent of the entry-wise absolute value matrix $|\gamma(\tfrac12)|$ grows
as $\sim 1.32^n$, while the actual permanent decays as $\sim e^{-n}$.
The complex phases therefore cancel all but a fraction $e^{-n}/1.32^n
\approx (0.28)^n$ of the magnitude --- roughly $n$ decades of
cancellation for moderate~$n$.
\end{remark}

\begin{remark}[Real-valuedness]
The permanent $\perm(\gamma(t))$ is real for all $t\in[0,1]$
(numerically, $|\operatorname{Im}\perm|/|\operatorname{Re}\perm|<10^{-14}$).
This follows from the circulant symmetry: $\gamma(t)$ is a circulant
with first row related by complex conjugation to its last row
(reversed), so the Ryser sum is invariant under complex conjugation.
\end{remark}

\subsection{Geodesic to the DFT matrix}\label{sec:dft-geodesic}

The normalized DFT matrix $F_n=S_n/\sqrt{n}$ is unitary.
The geodesic $I\to F_n$ also passes through a valley, but unlike
the cycle case, the behavior at $t=1$ depends strongly on the
arithmetic of~$n$.

\begin{observation}[Prime vs.\ composite recovery]\label{obs:prime}
Let \[r(n)=|\perm(F_n)|/\min_{t\in[0,1]}|\perm(\gamma(t))|\] and measure
how much the DFT permanent recovers above the geodesic valley.
Then:
\begin{itemize}
\item For odd prime~$n$: $r(n)\in[14,\,161]$ (recovery of $1$--$2$ decades).
\item For composite~$n$ ($n=9,15$): $r(n)\in[1.1,\,2.7]$ (negligible recovery).
\end{itemize}
\end{observation}

This provides a ``geodesic fingerprint'' of primality: the
number-theoretic structure of the DFT organizes the permanent
upward precisely when~$n$ is prime.

\subsection{Complex trajectories}\label{sec:complex-traj}

Figure~\ref{fig:geodesic-cx} shows the complex trajectory
$\perm(\gamma(t))$ for both geodesics.  For the cycle geodesic
(bottom row), the permanent remains real-valued, decaying
symmetrically from $\perm=1$ at $t=0,1$ to the midpoint minimum,
with the sign pattern $(-1)^{(n-1)/2}$ visible as
the alternation of $+$ and $-$ regions across different~$n$.
For the DFT geodesic (top row), the trajectory sweeps out an
increasingly circular arc in the complex plane: starting at $\perm=1$,
the permanent acquires a large imaginary part before returning
to a tiny real value at $t=1$.  The near-circularity of this arc for
large~$n$ is striking and unexplained.

\begin{figure}[ht]
\centering
\includegraphics[width=\textwidth]{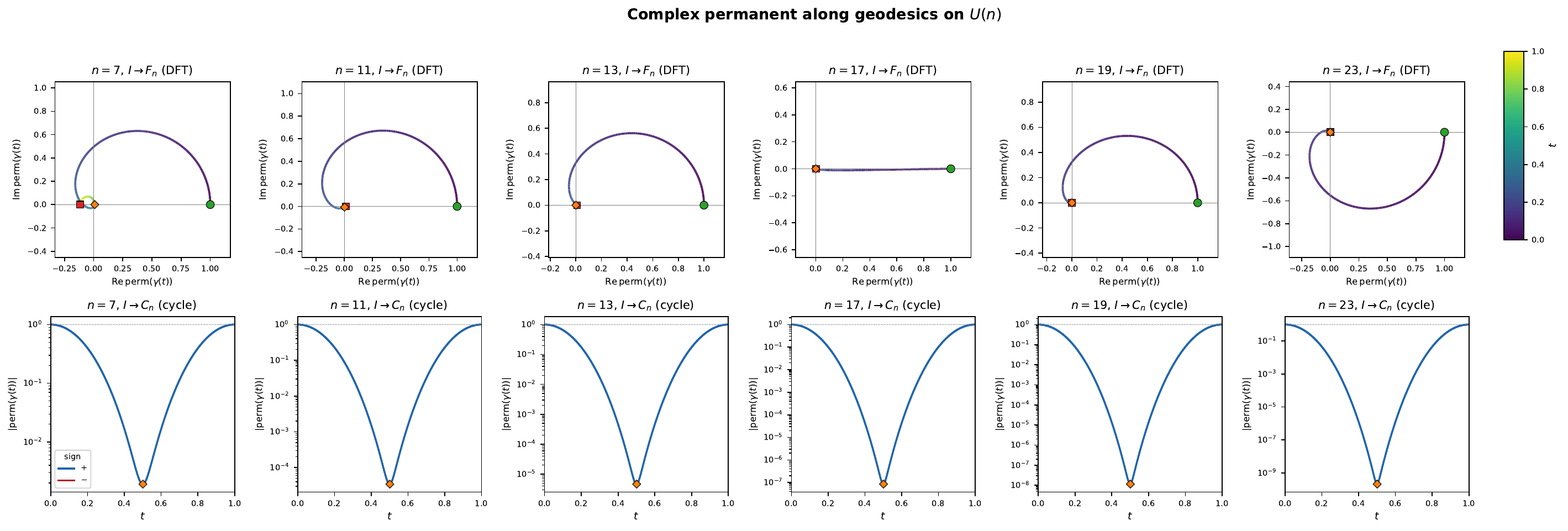}
\caption{Complex trajectory of $\perm(\gamma(t))$ along geodesics on
$U(n)$ for $n=7,11,13,17,19,23$.  Top: DFT geodesic $I\to F_n$, colored
by~$t$; green circle marks $t=0$, red square marks $t=1$, orange diamond
marks $t=1/2$.
Bottom: cycle geodesic $I\to C_n$ (log scale); the permanent is
real-valued, with blue/red indicating sign.}\label{fig:geodesic-cx}
\end{figure}

\section{Orthogonal matrices: a real counterpart}\label{sec:orthogonal}

We repeat the distributional analysis for the orthogonal group.
Since $O\in O(n)$ is real, $\perm(O)$ is a \emph{real} random
variable.  If the complex Gaussian for unitaries arose purely from a
CLT mechanism, we would expect $\perm(O)\sim N(0,\sigma^2)$ for
Haar-random $O$.

\subsection{Approximate normality with excess kurtosis}

For each $n\in\{7,11,13,17,19,23\}$, we generate $50{,}000$
Haar-random orthogonal matrices and compute $\perm(O)$ using the same
GPU pipeline.  The distribution is symmetric and bell-shaped
(Figure~\ref{fig:orthogonal}), but normality tests reveal systematic
\emph{excess kurtosis} that decreases with~$n$:

\begin{table}[ht]
\centering
\caption{Excess kurtosis of $\perm(O)$ for Haar-random
$O\in O(n)$.}\label{tab:orthogonal}
\medskip
\begin{tabular}{rccc}
\toprule
$n$ & $\operatorname{Var}[\perm(O)]$ & Excess kurtosis & SW $p$-value \\
\midrule
7  & $7.89\times 10^{-4}$ & $0.607$ & $<10^{-4}$ \\
11 & $3.88\times 10^{-6}$ & $0.418$ & $<10^{-4}$ \\
13 & $2.60\times 10^{-7}$ & $0.212$ & $0.001$ \\
17 & $1.14\times 10^{-9}$ & $0.145$ & $0.12$ \\
19 & $7.61\times 10^{-11}$ & $0.083$ & $0.19$ \\
\bottomrule
\end{tabular}
\end{table}

\begin{observation}\label{obs:orthogonal}
For Haar-distributed $O\in O(n)$, the permanent $\perm(O)$ is
approximately Gaussian with mean zero, but with positive excess
kurtosis that decreases as $O(1/n)$.  The convergence to Gaussian is
markedly slower than in the unitary case, where all normality tests
pass already at $n=7$.
\end{observation}

The positive excess kurtosis (leptokurtosis) means heavier tails than
Gaussian --- the orthogonal permanent produces occasional large values
more often than a Gaussian would predict.  This is consistent with the
heuristic that the orthogonal group has $n(n{-}1)/2$ degrees of freedom
(versus~$n^2$ for the unitary group), so the CLT-like cancellation
among the $n!$ Ryser terms proceeds more slowly.

\begin{figure}[ht]
\centering
\includegraphics[width=\textwidth]{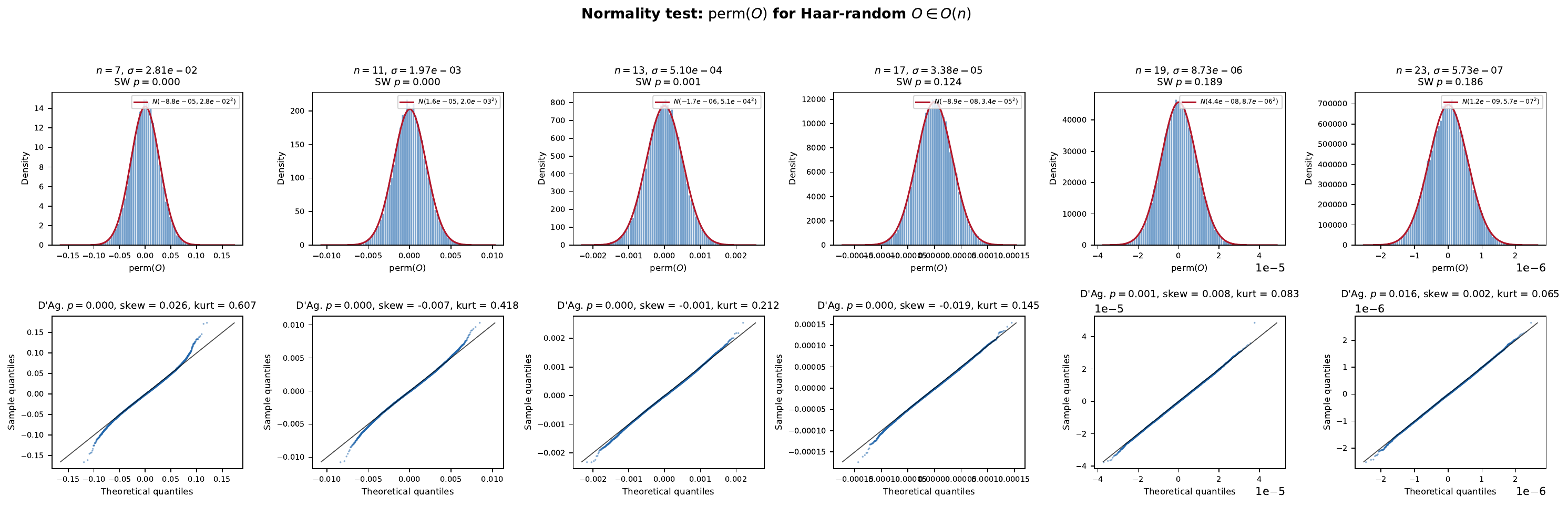}
\caption{Normality test for $\perm(O)$ with Haar-random $O\in O(n)$
($50{,}000$ samples each).  Top: histograms with fitted normal density.
Bottom: Q--Q plots.  Excess kurtosis is visible in the heavy tails for
small~$n$ but diminishes with increasing~$n$.}\label{fig:orthogonal}
\end{figure}

\subsection{Variance scaling}

The variance $\sigma^2(n)=\operatorname{Var}[\perm(O)]$ decays
exponentially, but significantly \emph{faster} than the unitary second
moment $n!/n^n\sim\sqrt{2\pi n}\,e^{-n}$ (Figure~\ref{fig:orth-var}).
The ratio
$\operatorname{Var}[\perm(O)]/(n!/n^n)$ decreases from $0.13$ at
$n=7$ to $0.001$ at $n=19$.  This discrepancy arises because the
orthogonal Weingarten function differs from its unitary counterpart:
the real entries of $O\in O(n)$ have stronger correlations (with
$n(n{-}1)/2$ rather than $n^2$ degrees of freedom), which amplifies
the cancellation in the Ryser sum and reduces the variance.

\begin{figure}[ht]
\centering
\includegraphics[width=0.5\textwidth]{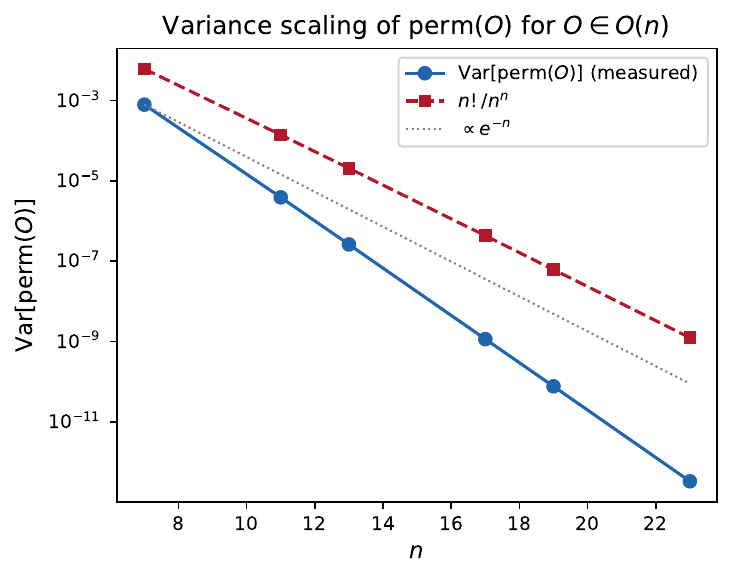}
\caption{Variance of $\perm(O)$ for Haar-random $O\in O(n)$,
compared with the unitary second moment $n!/n^n$.  The orthogonal
variance decays faster.}\label{fig:orth-var}
\end{figure}

\begin{remark}[Cycle geodesic on $O(n)$]
Since the $n$-cycle permutation matrix $C_n$ lies in $O(n)$ (and in
$SO(n)$ for odd~$n$), the geodesic $I\to C_n$ on $U(n)$ stays
entirely within~$O(n)$: the eigenvalues of $C_n$ come in conjugate
pairs, so the geodesic matrices are real.  All the cycle-geodesic
results of Section~\ref{sec:cycle} therefore apply equally to
$O(n)$.  The distinction between orthogonal and unitary permanents
manifests only in the \emph{distributional} behavior, not the
geodesic geometry.
\end{remark}

\section{Gaussian matrix ensembles: heavy-tailed permanents}
\label{sec:gaussian}

The complex Gaussian distribution of $\perm(U)$ for Haar unitaries
relies on the CLT heuristic: the permanent is a sum of $n!$ terms,
each of magnitude $\sim n^{-n/2}$, and if these are sufficiently
decorrelated, a Gaussian limit follows.  A natural question is whether
this mechanism extends to other natural matrix ensembles.  We
investigate four Gaussian ensembles:
\begin{itemize}[leftmargin=2em]
\item \textbf{GUE}: $H=(A+A^*)/2$, where $A$ has i.i.d.\
  $\mathcal{CN}(0,1)$ entries;
\item \textbf{GOE}: $H=(A+A^T)/2$, where $A$ has i.i.d.\
  $N(0,1)$ entries;
\item \textbf{Complex Ginibre}: $G$ with i.i.d.\
  $\mathcal{CN}(0,1)$ entries;
\item \textbf{Real Ginibre}: $G$ with i.i.d.\ $N(0,1)$ entries.
\end{itemize}
For each ensemble, we compute $20{,}000$ permanents at each
$n\in\{7,11,13,17\}$ using the GPU pipeline.

\subsection{Permanents of Hermitian matrices are real}

\begin{proposition}\label{prop:herm-real}
If $H$ is an $n\times n$ Hermitian matrix, then $\perm(H)\in\mathbb{R}$.
\end{proposition}

\begin{proof}
We have $\overline{h_{ij}}=h_{ji}$, so
\[
  \overline{\perm(H)}
  = \sum_{\sigma\in S_n}\prod_{i=1}^n \overline{h_{i,\sigma(i)}}
  = \sum_{\sigma\in S_n}\prod_{i=1}^n h_{\sigma(i),i}
  = \sum_{\sigma^{-1}\in S_n}\prod_{j=1}^n h_{j,\sigma^{-1}(j)}
  = \perm(H),
\]
where we used the bijection $\sigma\mapsto\sigma^{-1}$ on $S_n$.
\end{proof}

In particular, $\perm(H)$ is real for every GUE matrix, despite $H$
having complex entries.  This is confirmed numerically:
$|\operatorname{Im}\perm(H)|/|\operatorname{Re}\perm(H)|<10^{-14}$
across all $80{,}000$ samples.

\subsection{Failure of the CLT: $\alpha$-stable distributions}

The permanent distributions for all four ensembles are spectacularly
non-Gaussian: excess kurtosis ranges from $34$ (GUE, $n=7$) to $689$
(GOE, $n=17$) and \emph{increases} with~$n$, opposite to the
convergence seen for Haar unitaries.  Neither Gaussian nor Student-$t$
distributions provide adequate fits (KS $p$-values $<10^{-4}$ for
Gaussian; Student-$t$ with fitted degrees of freedom
$\nu\approx 1$--$2.5$ is better but still rejected).

However, the permanents are well described by \emph{$\alpha$-stable
distributions}~\cite{SamorodNolan}.  A symmetric $\alpha$-stable
random variable $X\sim S(\alpha,0,\gamma,\delta)$ is characterized by
its stability index $\alpha\in(0,2]$: the case $\alpha=2$ is
Gaussian, $\alpha=1$ is Cauchy, and for $\alpha<2$ the tails decay as
$P(|X|>x)\sim x^{-\alpha}$ (so $\operatorname{Var}[X]=\infty$ for
$\alpha<2$ and $\mathbb{E}[|X|]=\infty$ for $\alpha\le 1$).  We
estimate~$\alpha$ using both McCulloch's quantile
method~\cite{McCulloch86} and regression on the empirical
characteristic function; the two methods agree to within $\pm 0.03$.

\begin{observation}\label{obs:stable}
The permanent of an $n\times n$ matrix drawn from each Gaussian
ensemble follows a symmetric $\alpha$-stable distribution with
stability index~$\alpha$ approximately constant in~$n$:
\begin{center}
\begin{tabular}{lcccc}
\toprule
 & $n=7$ & $n=11$ & $n=13$ & $n=17$ \\
\midrule
GOE             & $1.02$ & $0.99$ & $0.98$ & $0.96$ \\
GUE             & $1.21$ & $1.16$ & $1.18$ & $1.15$ \\
Ginibre (real)  & $1.19$ & $1.17$ & $1.18$ & $1.19$ \\
Ginibre (complex) & $1.42$ & $1.39$ & $1.39$ & $1.36$ \\
\bottomrule
\end{tabular}
\end{center}
In all cases $|\beta|<0.06$ (symmetric), and KS $p$-values for the
stable fit range from $0.08$ to $0.98$.
\end{observation}

The hierarchy of stability indices is:
\[
  \alpha_{\mathrm{GOE}}\approx 1.0
  \;<\;
  \alpha_{\mathrm{GUE}}\approx\alpha_{\text{Gin-R}}\approx 1.2
  \;<\;
  \alpha_{\text{Gin-C}}\approx 1.4
  \;<\;
  \alpha_{\mathrm{Haar}}\,{=}\,2.
\]
The GOE permanent is essentially Cauchy-distributed ($\alpha\approx 1$,
infinite mean), while the complex Ginibre permanent has the lightest
tails among the Gaussian ensembles, yet is still far from Gaussian.

Figure~\ref{fig:stable} shows the complementary CDF on log-log axes
for $n=13$ (left panel), where the power-law tails are clearly visible
and the Gaussian reference drops off too steeply, together with the
fitted $\alpha$ as a function of~$n$ (right panel), confirming the
stability of $\alpha$ across matrix sizes.
Figure~\ref{fig:stable-qq} presents Q--Q plots of the data against
fitted stable quantiles for all four ensembles and all four values
of~$n$.  The agreement along the diagonal confirms the quality of the
stable fit across the entire range, including the extreme tails.

\begin{figure}[ht]
\centering
\includegraphics[width=\textwidth]{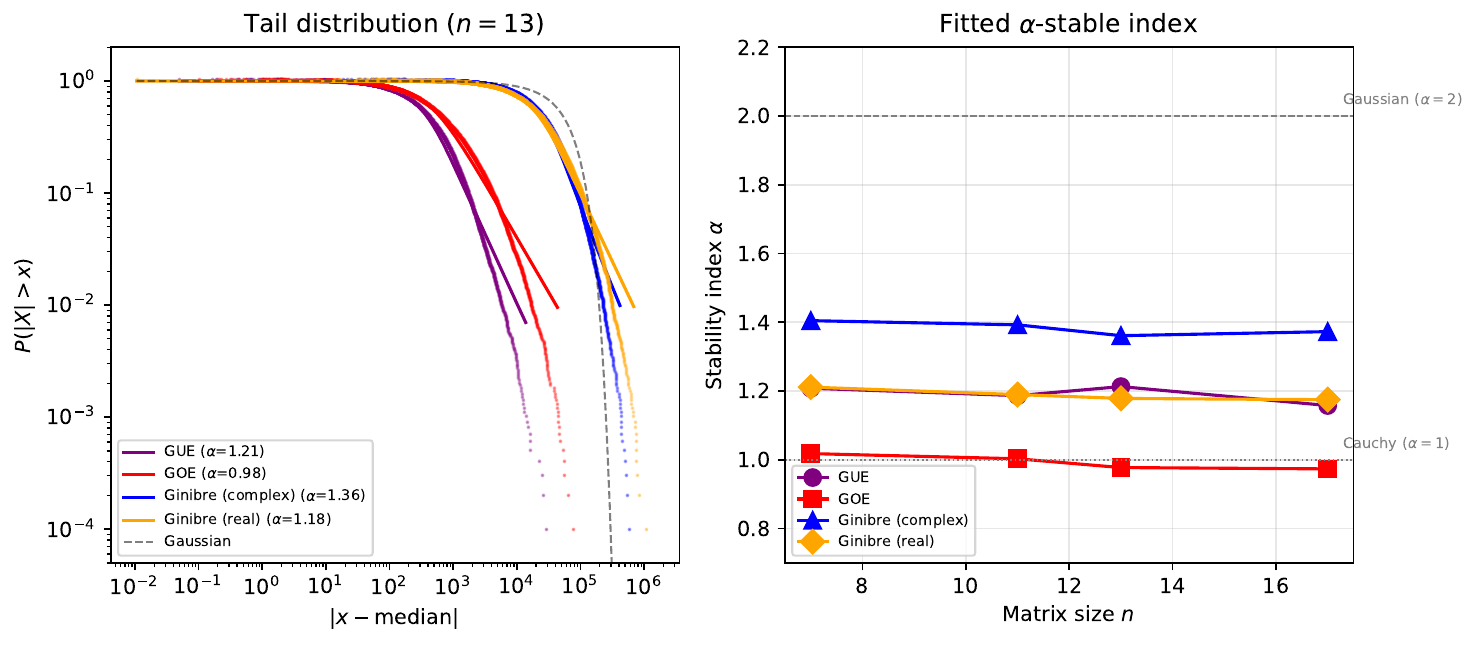}
\caption{Left: complementary CDF (tail distribution) of $|\perm - \mathrm{median}|$
for $n=13$ across four Gaussian ensembles, with fitted stable
distributions (solid curves) and Gaussian reference (dashed).
Right: fitted stability index $\alpha$ vs.\ matrix size $n$;
$\alpha=2$ (Gaussian) and $\alpha=1$ (Cauchy) reference lines are
shown.}\label{fig:stable}
\end{figure}

\begin{figure}[ht]
\centering
\includegraphics[width=\textwidth]{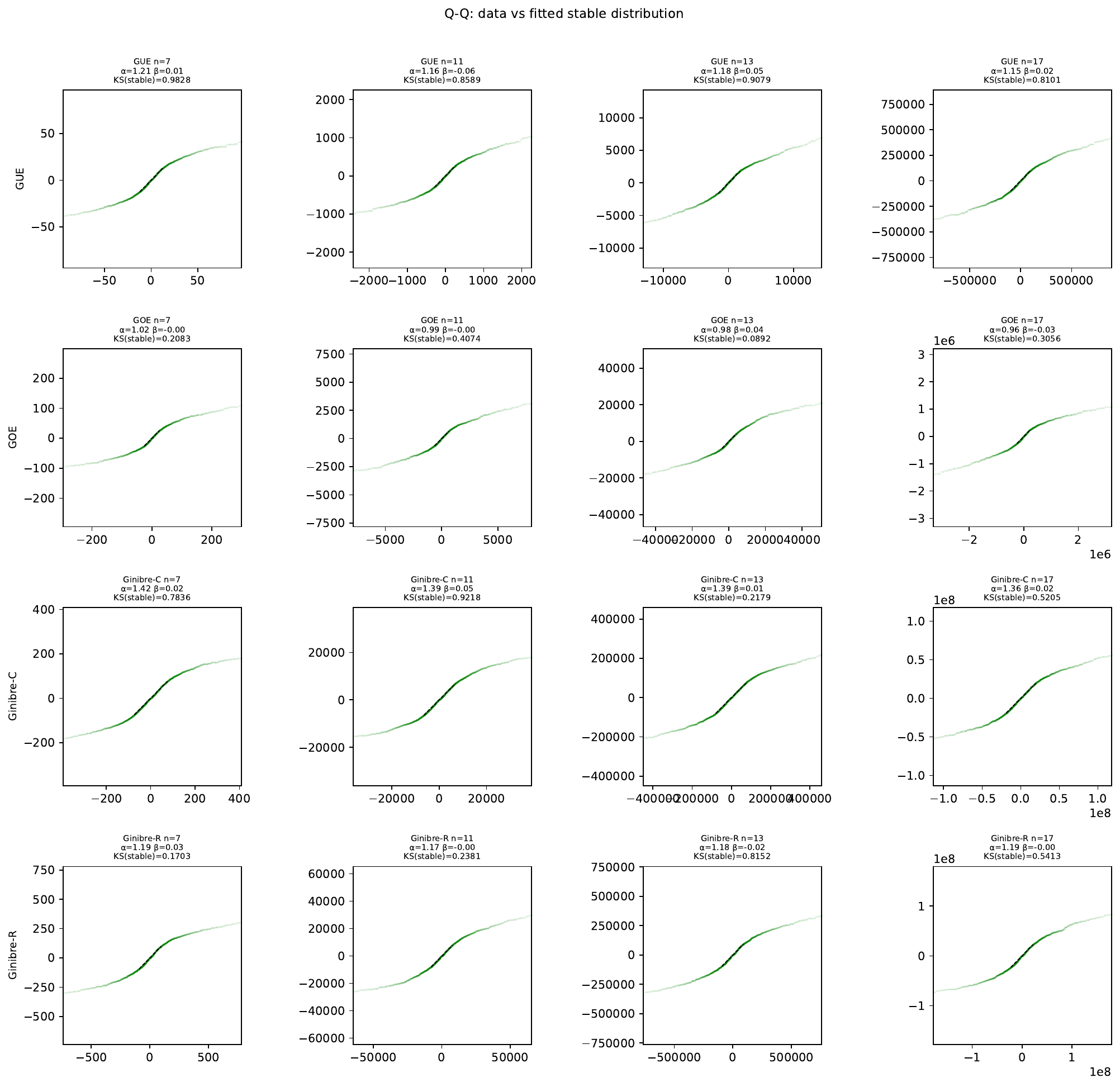}
\caption{Q--Q plots of permanent data (vertical axis) against fitted
$\alpha$-stable quantiles (horizontal axis) for four Gaussian matrix
ensembles and $n=7,11,13,17$ ($20{,}000$ samples each).  Each panel
shows the fitted $\alpha$, $\beta$, and KS $p$-value.  Points falling
on the diagonal indicate a good fit.}\label{fig:stable-qq}
\end{figure}

\begin{remark}[Unitarity is essential for Gaussianity]
The key difference between Haar unitaries (where $\perm(U)$ is
Gaussian) and the Gaussian ensembles (where it is $\alpha$-stable
with $\alpha<2$) is that unitary matrices have \emph{bounded} entries
($|U_{ij}|\le 1$), while Gaussian matrices have unbounded entries.
When all entries are $O(n^{-1/2})$, the $n!$ terms in the Ryser sum
are comparable in magnitude, producing good CLT-like cancellation.
For Gaussian matrices, occasional large entries create heavy-tailed
products that dominate the sum, destroying the CLT mechanism and
producing stable (rather than Gaussian) limits.
\end{remark}

\section{Aaronson's lognormal conjecture}\label{sec:lognormal}

Aaronson and Arkhipov~\cite{AA14} proved that for the complex Ginibre
ensemble ($X$ with i.i.d.\ $\mathcal{CN}(0,1)$ entries),
$|\det(X)|^2$ converges to a lognormal random variable, and
conjectured the same for $|\perm(X)|^2$ based on the visual similarity
of their densities.  We test this conjecture numerically for all four
Gaussian ensembles.

\subsection{Setup}

For each ensemble (complex Ginibre, real Ginibre, GUE, GOE), we
generate $50{,}000$ matrices at each~$n$ and compute $\log|\cdot|$.
Determinants are computed via LU decomposition for
$n\in\{7,9,\ldots,101\}$; permanents via GPU-accelerated Ryser for
$n\le 17$.  We test normality of $\log|\cdot|$ via
Kolmogorov--Smirnov, Shapiro--Wilk, Anderson--Darling, and
D'Agostino--Pearson tests, plus skewness and excess kurtosis.

\subsection{Results}

Table~\ref{tab:lognormal} compares $\log|\perm|$ and $\log|\det|$.

\begin{table}[ht]
\centering
\caption{Normality of $\log|\perm|$ vs.\ $\log|\det|$ ($50{,}000$
samples).}\label{tab:lognormal}
\medskip
\begin{tabular}{llrrrrrr}
\toprule
& & \multicolumn{3}{c}{$\log|\det|$} & \multicolumn{3}{c}{$\log|\perm|$} \\
\cmidrule(lr){3-5}\cmidrule(lr){6-8}
Ensemble & $n$ & Skew & Kurt & AD & Skew & Kurt & AD \\
\midrule
Ginibre-C & 7  & $-0.48$ & $0.61$ & $111$ & $-0.48$ & $0.68$ & $109$ \\
          & 13 & $-0.38$ & $0.39$ & $73$  & $-0.39$ & $0.65$ & $68$ \\
          & 17 & $-0.34$ & $0.34$ & $55$  & $-0.39$ & $0.63$ & $67$ \\
\midrule
GOE       & 7  & $-0.63$ & $0.79$ & $195$ & $-0.58$ & $0.87$ & $154$ \\
          & 13 & $-0.50$ & $0.60$ & $126$ & $-0.47$ & $0.84$ & $100$ \\
          & 17 & $-0.44$ & $0.46$ & $97$  & $-0.46$ & $0.93$ & $87$ \\
\midrule
Ginibre-R & 7  & $-0.74$ & $1.34$ & $249$ & $-0.77$ & $1.37$ & $271$ \\
          & 13 & $-0.66$ & $1.10$ & $189$ & $-0.71$ & $1.37$ & $217$ \\
          & 17 & $-0.58$ & $0.89$ & $147$ & $-0.77$ & $1.86$ & $227$ \\
\midrule
GUE       & 7  & $-0.76$ & $1.21$ & $280$ & $-0.80$ & $1.56$ & $285$ \\
          & 13 & $-0.63$ & $1.05$ & $171$ & $-0.69$ & $1.42$ & $210$ \\
          & 17 & $-0.57$ & $0.87$ & $148$ & $-0.73$ & $1.86$ & $211$ \\
\bottomrule
\end{tabular}
\end{table}

Figure~\ref{fig:lognormal-qq} presents Q--Q plots comparing
$\log|\det|$ and $\log|\perm|$ against normal quantiles for the
complex Ginibre ensemble, and Figure~\ref{fig:lognormal-stats}
summarizes the skewness and kurtosis trajectories across all ensembles.

\begin{figure}[ht]
\centering
\includegraphics[width=\textwidth]{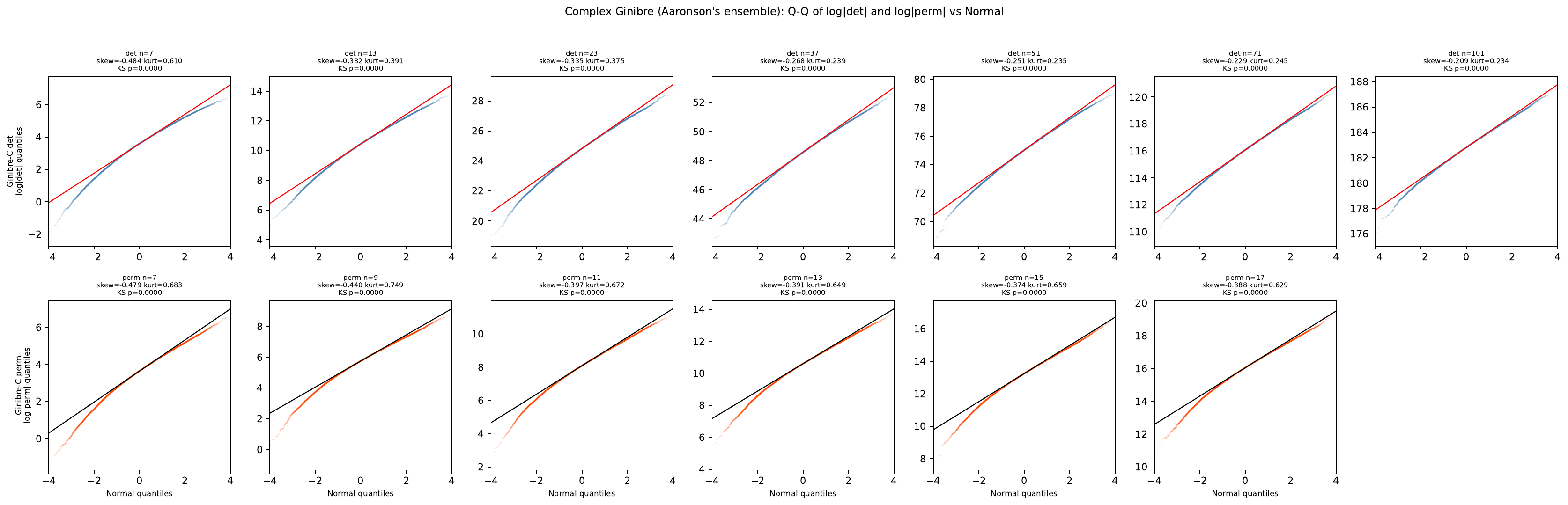}
\caption{Q--Q plots of $\log|\det|$ (top, blue) and $\log|\perm|$
(bottom, red) against normal quantiles for the complex Ginibre ensemble.
The determinant straightens visibly as $n$ grows from $7$ to $101$;
the permanent improves more slowly over its accessible range $n\le 17$.
}\label{fig:lognormal-qq}
\end{figure}

\begin{figure}[ht]
\centering
\includegraphics[width=\textwidth]{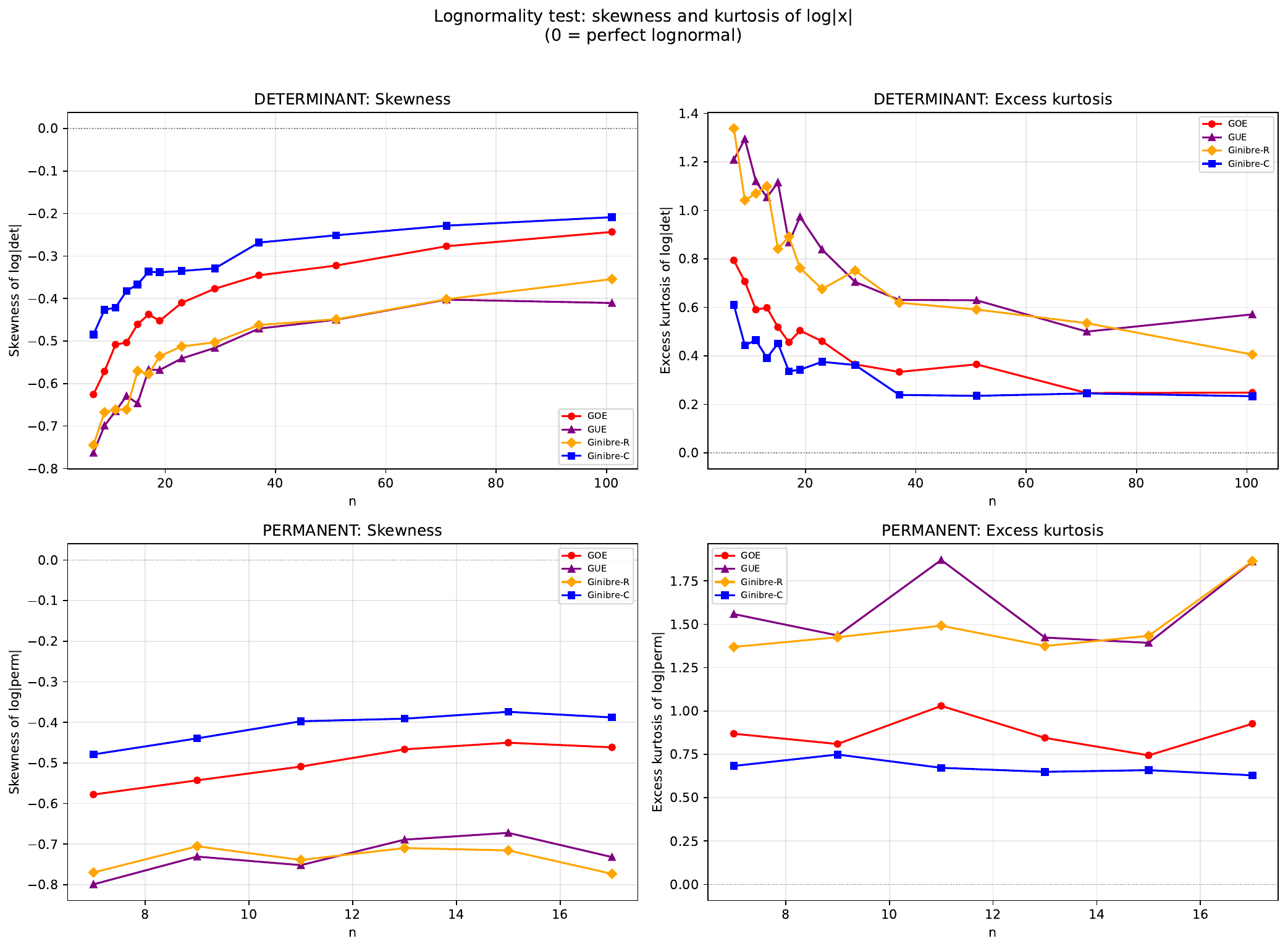}
\caption{Skewness and excess kurtosis of $\log|x|$ for determinants
(top) and permanents (bottom) across all four Gaussian ensembles.
Normal corresponds to $(0,0)$.  The determinant converges toward
$(0,0)$ for all ensembles; the permanent converges for GOE and
Ginibre-C but not for GUE and Ginibre-R.}\label{fig:lognormal-stats}
\end{figure}

A striking dichotomy emerges:
\begin{itemize}
\item \textbf{Ginibre-C and GOE}: The permanent tracks the
  determinant closely; both improve with~$n$.  Aaronson's lognormal
  conjecture is plausible for these ensembles.
\item \textbf{GUE and Ginibre-R}: The permanent \emph{diverges} from
  the determinant.  While $\log|\det|$ steadily improves,
  $\log|\perm|$ shows flat or worsening skewness ($\approx{-}0.7$ to
  ${-}0.8$) and kurtosis ($\approx 1.4$ to $1.9$).  A lognormal
  conjecture for these ensembles appears to be false.
\end{itemize}

\subsection{Connection to $\alpha$-stable distributions}

The divergent behavior of GUE/Ginibre-R permanents has a natural
explanation via the $\alpha$-stable distributions of
Section~\ref{sec:gaussian}.  For $\alpha<2$, the power-law tails
$P(|X|>x)\sim x^{-\alpha}$ produce a heavy \emph{left} tail in
$\log|X|$ (many near-zero values), manifesting as persistent negative
skewness that does not diminish with~$n$.  The hierarchy is:
$\alpha\approx 1.4$ (Ginibre-C, lightest tails, lognormal most
plausible) $>$ $\alpha\approx 1.2$ (GUE, Ginibre-R, lognormal
doubtful) $>$ $\alpha\approx 1.0$ (GOE, Cauchy-like, yet $\log|\perm|$
still improves slowly).  The determinant converges to lognormal for
\emph{all} ensembles because $\log|\det(X)|^2=\sum_k\log\chi^2_{2k}$
is a sum of independent random variables to which Berry--Esseen applies.

\subsection{Anti-concentration for Gaussian ensembles}\label{sec:ac}

The PACC (see Section~\ref{sec:boson}) asks whether
$\Pr[|\perm|^2\ge\mathbb{E}[|\perm|^2]/\mathrm{poly}(n)]\ge
1/\mathrm{poly}(n)$.  For Gaussian matrices,
$\mathbb{E}[|\perm(X)|^2]=n!$ by independence of entries.

\begin{observation}\label{obs:ac}
Anti-concentration holds for all four Gaussian ensembles, and is in
fact \textbf{more robust} than for Haar unitaries: at the threshold
$\Pr[|\perm|^2\ge\mathbb{E}/n^2]$, the Gaussian ensembles
show \emph{increasing} probability with~$n$ ($67\%$--$98\%$ at
$n=17$), while Haar unitaries \emph{decrease} ($18\%$ at $n=17$).
\end{observation}

The heavy $\alpha$-stable tails that destroy lognormality
simultaneously \emph{strengthen} anti-concentration: a symmetric
$\alpha$-stable distribution with $\alpha>1$ has bounded density at the
origin, preventing mass from piling up near zero.  The failure of
the CLT and the success of anti-concentration are two sides of the
same coin (Figure~\ref{fig:ac}).

\begin{figure}[ht]
\centering
\includegraphics[width=\textwidth]{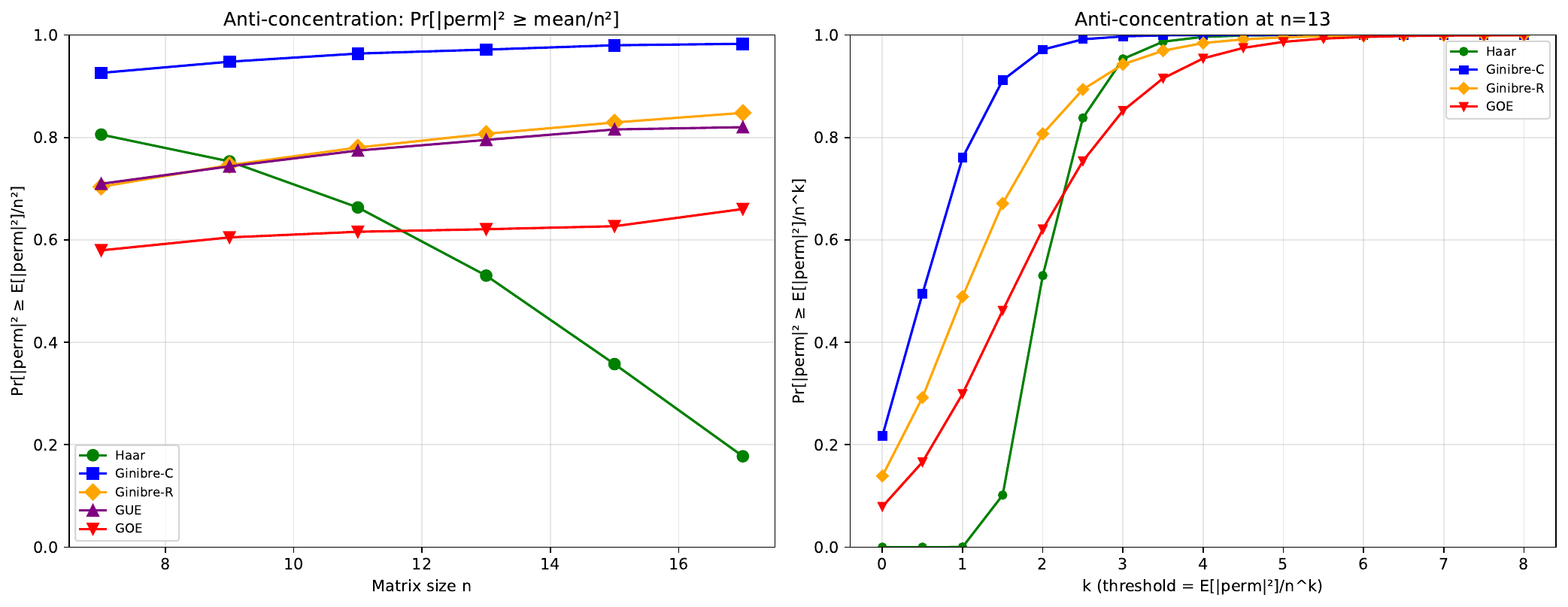}
\caption{Left: $\Pr[|\perm|^2\ge\mathbb{E}[|\perm|^2]/n^2]$ vs.\
matrix size~$n$ for all ensembles.  The Gaussian ensembles
\emph{increase} with~$n$ while Haar unitaries \emph{decrease}.
Right: anti-concentration curve at $n=13$, showing
$\Pr[|\perm|^2\ge\mathbb{E}/n^k]$ as a function of~$k$.}\label{fig:ac}
\end{figure}

\section{Discussion}\label{sec:discussion}

\subsection{Relation to prior work}

The complex Gaussian observation (Observation~\ref{obs:rayleigh})
complements work on the moments and anti-concentration of unitary
permanents.
Nezami~\cite{Nezami21} computed exact moments
$\mathbb{E}[|\perm(U)|^{2k}]$ via representation theory and
established an asymptotic expansion.
The Permanent Anti-Concentration Conjecture (PACC), central to the
hardness of boson sampling, asserts that
$\Pr[|\perm(U)|^2\ge n!/\mathrm{poly}(n)]\ge 1/\mathrm{poly}(n)$
for Haar-random $U$; recent progress includes
\cite{FeffermanGhoshZhan24,KwanSahSawhney24}.
Our complex Gaussian observation identifies the \emph{full
distributional form} of $\perm(U)$, not just tail bounds.  If proven,
it would imply sharp anti-concentration constants.  The $\alpha$-stable
distributions for Gaussian ensembles (Section~\ref{sec:gaussian})
appear to be entirely new and suggest that the Gaussianity of
$\perm(U)$ for Haar unitaries is a delicate consequence of the
boundedness of unitary matrix entries.

The geodesic results (Observations~\ref{obs:universal}--\ref{obs:prime})
appear to be entirely new.  The midpoint formula
$\perm(\gamma(\frac12))=(-1)^{(n-1)/2}\cdot 2e^{-n}(1+\frac{1}{3n}+\cdots)$
is supported by high-precision numerics and likely admits a proof via
saddle-point analysis of the Ryser formula applied to the explicit
circulant~\eqref{eq:circulant}.

\subsection{The universal function $f(t)$}

The function $f(t)=-\frac{1}{n}\ln|\perm(\gamma(t))|$ for the
$n$-cycle geodesic appears to satisfy:
\begin{itemize}
\item $f(t)=f(1-t)$ (from the symmetry $\perm(\gamma(t))=\overline{\perm(\gamma(1-t))}$);
\item $f(t)\sim \frac{\pi^2}{3}t^2$ as $t\to 0^+$;
\item $f(\frac12)=1$ (the $2e^{-n}$ asymptotic).
\end{itemize}
Finding a closed form for $f(t)$ is an attractive open problem.
The initial quadratic $\pi^2 t^2/3$ corresponds to a Gaussian decay
$|\perm(\gamma(t))|\approx e^{-n\pi^2 t^2/3}$ near the identity,
suggesting a connection to the spectral properties of the
skew-Hermitian generator $\log C_n$.

\subsection{Connections to boson sampling}\label{sec:boson}

Our results bear on several aspects of the boson sampling
program~\cite{AaronsonArkhipov13}.

\medskip\noindent\textbf{Porter--Thomas statistics.}
In an $n$-photon linear optical network described by a Haar-random
unitary~$U$, the probability of detecting one photon in each of $n$
specified output modes is $p(T)=|\perm(U_T)|^2/(n_1!\cdots n_m!)$,
where $U_T$ is the corresponding $n\times n$ submatrix.  For the
collision-free case ($n_i\in\{0,1\}$), this reduces to
$p(T)=|\perm(U_T)|^2$.  The same CLT heuristic that drives the full
permanent toward $\mathcal{CN}(0,\sigma^2)$ applies to submatrix
permanents, and the exponential law~\eqref{eq:exp} therefore predicts
that the output probabilities follow Porter--Thomas
statistics~\cite{PorterThomas56}: an exponential distribution with
mean $n!/n^n$ (for $n$-mode submatrices of a Haar-random $n\times n$
unitary).  This has been verified experimentally in photonic boson
sampling demonstrations~\cite{Zhong2020}.

\medskip\noindent\textbf{Anti-concentration.}
The Permanent Anti-Concentration Conjecture (PACC), central to the
hardness of approximate boson
sampling~\cite{AaronsonArkhipov13,FeffermanGhoshZhan24}, asserts that
\[
  \Pr\bigl[|\perm(U)|^2 \ge n!/\mathrm{poly}(n)\bigr]
  \;\ge\; 1/\mathrm{poly}(n).
\]
The exponential form~\eqref{eq:exp} gives the \emph{much} stronger
\[
  \Pr\bigl[|\perm(U)|^2 \ge \sigma^2/\mathrm{poly}(n)\bigr]
  = e^{-1/\mathrm{poly}(n)}
  = 1-1/\mathrm{poly}(n),
\]
so \emph{almost all} Haar-random permanents exceed the PACC threshold.
Proving the complex Gaussian distributional form would therefore
resolve the PACC with room to spare.  Even short of a full proof,
establishing that $\perm(U)$ is sub-Gaussian (which requires only
controlling the moment-generating function) would suffice.

\medskip\noindent\textbf{DFT interferometers.}
The DFT outlier phenomenon (Observation~\ref{obs:dft}) is relevant to
experiments using Fourier multi-port
interferometers~\cite{Crespi2016}, where the large permanent of the
DFT matrix suggests that certain transition amplitudes are anomalously
large compared to generic unitaries.
This complements the suppression laws of
Tichy~\cite{Tichy2014}, which identify \emph{vanishing}
submatrix permanents of the DFT.

\medskip\noindent\textbf{Classical simulation frontier.}
From a computational perspective, the fastest known classical algorithm
for exact boson sampling probabilities requires computing the
permanent in $O(n\,2^n)$ time~\cite{CliffordClifford18,Ryser63}.
Our GPU-accelerated pipeline demonstrates that pushing exact permanent
computation to $n\approx 43$ is now feasible, which overlaps with the
regime of current photonic experiments~\cite{Zhong2020}.

\subsection{Open problems}

\begin{enumerate}
\item Prove that $\perm(U)$ converges in distribution to a
  circularly-symmetric complex Gaussian (after appropriate scaling) as
  $n\to\infty$.  For the orthogonal case, determine the rate at which
  the excess kurtosis of $\perm(O)$ vanishes.
\item Prove the midpoint formula
  $\perm(\gamma(\frac12))=(-1)^{(n-1)/2}\cdot 2e^{-n}(1+\frac{1}{3n}+O(n^{-2}))$.
\item Find a closed form for the universal function~$f(t)$.
\item Explain the prime-vs-composite dichotomy for the DFT geodesic
  (Observation~\ref{obs:prime}).
\item Prove that the permanent of a GOE matrix converges in
  distribution (after centering) to a Cauchy law, and determine
  $\alpha$ analytically for the other Gaussian ensembles.
\item Prove or disprove Aaronson's lognormal conjecture for the
  complex Ginibre ensemble.  For GUE and real Ginibre, our evidence
  suggests the conjecture is false.
\item Extend the table of $\perm(S_n)$ further; $n=45$ is within
  reach ($\sim\!28$ GPU-hours).
\end{enumerate}


\end{document}